\documentclass[a4paper,aps,pra,showpacs,twocolumn,superscriptaddress]{revtex4-1}   
   
\usepackage[utf8]{inputenc}  
\usepackage[T1]{fontenc}     
\usepackage[british]{babel}  
\usepackage{lmodern}  
\usepackage[scaled=1.03]{inconsolata} 
\usepackage[usenames,dvipsnames]{color} 
\usepackage[colorlinks,citecolor=blue,linkcolor=magenta,urlcolor=blue]{hyperref}  
\usepackage{graphicx} 
\usepackage{tikz}
\usepackage[babel]{microtype}  
\usepackage{amsmath,amssymb,amsthm,bm,mathtools,amsfonts,mathrsfs,bbm,dsfont} 
\usepackage{xspace}  
\usepackage{multirow}
\usepackage{verbatim}
\usepackage{float}

\usepackage{physics}
\newcommand{\id}{\ensuremath{\mathds{1}}}
\usepackage{bbold}
\usepackage{tcolorbox}
\usepackage{nicefrac, xfrac}

\newcommand{\hil}{\mathcal{H}}

\newtheoremstyle{mystyle}
  {6pt}
  {6pt}
  {\normalfont}
  {0pt}
  {\bf}
  {.}
  { }
  {}

\theoremstyle{mystyle}

\newtheorem{definition}{Definition}
\newtheorem{corollary}{Corollary}
\newtheorem{observation}{Observation}

\usepackage{dblfloatfix}
\usepackage{blindtext}

\usepackage{xcolor}
\colorlet{myPurple}{blue!40!red}
\colorlet{myCyan}{cyan!50!gray}

\definecolor{quantumviolet}{HTML}{53257F} 
\definecolor{quantumgray}{HTML}{555555} 
\definecolor{mygray}{gray}{0.95} 

\newtcolorbox[auto counter,number within=section]{boxfigure}[2][]{%
colback=mygray,colframe=myPurple,fonttitle=\bfseries,width=\columnwidth,float*=ht,lower separated=false, halign=justify,title=Box~\thetcbcounter: #2,#1}

\usepackage{cases}

\begin{document}

\nonfrenchspacing

\title{Certification of Quantum Networks using the generalised Choi isomorphism}

\author{Sophie Egelhaaf}
\affiliation{Department of Applied Physics, University of Geneva, Switzerland}

\author{Roope Uola}
\affiliation{Department of Applied Physics, University of Geneva, Switzerland}
\affiliation{Department of Physics and Astronomy, Uppsala University, Box 516, 751 20 Uppsala, Sweden}
\affiliation{Nordita, KTH Royal Institute of Technology and Stockholm University, Hannes Alfvéns väg 12, 10691 Stockholm, Sweden}

\date{\today}  

\begin{abstract}
We present a framework for certifying entanglement properties of quantum states and measurements in line networks. The framework is based on the generalised Choi isomorphism, which can be used to map bipartite states and measurement into corresponding quantum operations. We apply the method to networks with trusted end points to demonstrate the power of the approach. We derive bounds for common convex geometric entanglement quantifiers of individual source states, as well as for the network as a whole. We also apply the technique to certification of entanglement dimensionality, proposing the concept of Schmidt number for bipartite measurements in the process. We believe this quantifier can find interest in benchmarking detectors used, e.g., in qutrit teleportation. Applying our formalism further to high-dimensional networks, we derive an activation result for genuinely high-dimensional quantum steering.
\end{abstract}

\maketitle

\section{Introduction}

Quantum networks form a rapidly developing research field with applications varying from quantum communication to foundational tests of quantum mechanics and even to the quantum internet \cite{Wei_2022,Tavakoli_2022,RevModPhys.95.045006}. Networks provide new challenges and opportunities for information processing, such as the need for developing novel methods for characterizing non-convex sets of correlations, as well as the possibility of extracting randomness without inputs.

With complex network structures, it becomes necessary to develop techniques for characterizing the quantum properties that a given network possesses. The aim can, for example, be to characterise the possible quantum advantage a given network can reach \cite{SekatskiBoreiri2023}, or the set of quantum states achievable with a given network structure \cite{Kraft_2021}. For practical purposes, it is imperative to certify that all the nodes and vertices of a network are functioning up to the standards required for, e.g., quantum communication.
 
In this work, we focus on the certification of networks that have untrusted nodes. This could be the case, for example, in a communication network, where some nodes are inaccessible to the user, or where some parties might wish to act maliciously. We focus mainly on line-shaped networks, where we have access to trusted endpoints or, equivalently, ring-shaped networks, where we have access to two adjacent trusted nodes. We show that this setup allows for qualitative and quantitative bounds on the key quantum properties, such as entanglement, present in the inaccessible vertices and nodes of the network.

Our methods are based on network quantum steering \cite{Jones_2021} and more specifically to the related generalised Choi isomorphism \cite{Shirokov2008,Kiukas2017}. Since we only trust the end nodes, our methods provide a semi-device independent (SDI) tool to certify various entanglement properties. Moreover, as networks have the mathematically strong assumption of independent sources, network steering provides an SDI certification method for objects that cannot be certified in the traditional steering scenario. Concretely, we provide lower bounds on common convex geometric quantifiers of entanglement. These include the robustness, as well as the convex weight, for individual states and measurements. We also discuss lower bounds for these quantifiers in the network as a whole. We further show that our framework can be applied to derive bounds on entanglement dimensionality, i.e. the Schmidt number, of the objects in the network. This leads us to propose a definition of Schmidt number for bipartite measurements, which may be of independent interest, especially for benchmarking current optical platforms capable of two-qutrit teleportation; see, e.g. Ref.~\cite{Luo_2019}. We end the paper by using our framework to show that some one-way steerable states can be activated in a bilocality network to allow for SDI Schmidt number certification.

\section{Viewing Networks as operations}

\subsection{Standard bipartite steering}

In standard quantum steering~\cite{Wiseman_2007,Cavalcanti2017,Uola2020} Alice and Bob share a bipartite state $\rho^{AB}$, i.e. a positive semi-definite operator with trace one acting on a tensor product of two Hilbert spaces. This work assumes the Hilbert spaces to be finite-dimensional.
Bob performs a set of local measurements $B_{b|y}^B$ on his part of the state, with $y$ labeling the measurement choice and $b$ the corresponding outcome. Here, we concentrate on scenarios with a finite number of measurements and outcomes.
In quantum mechanics, a measurement is described by a positive operator-valued measure (POVM), that is a collection of positive semi-definite matrices summing up to the identity, i.e. $\sum_b B_{b|y}=\id$.
After the state update caused by this measurement, Alice is left with a state assemblage, cf. Fig.~\ref{fig:twoParty},
\begin{align}\label{Eq:StateAssemblage}
    \sigma^A_{b|y}:=\tr_B(\rho^{AB}(\id^A \otimes B_{b|y}^B)).
\end{align}
Note that the members of the state assemblage are sub-normalised, i.e. their trace is less than or equal to one. 
In the standard formulation of finite-dimensional quantum steering, the question is whether the assemblage is preparable with some separable state \cite{Kogias_2015,Moroder_2016,jokinen2023compressingcontinuousvariablequantum}.
Here, no assumption about the Hilbert space dimension of Bob's system is made, but it is kept finite. Hence, typically steering is referred to as one-sided device-independent entanglement certification and Bob's measurements in such protocol are called untrusted.

The shared state $\rho^{AB}$ in a steering scenario can be viewed as an operation, which is a linear completely positive and trace non-increasing map. This operation maps sets of Bob's POVMs $B_{b|y}^B$ to state assemblages on Alice $\sigma^A_{b|y}$.
The connection is based on a generalized version of the Choi-Jamiolkowski isomorphism between quantum channels, i.e. trace-preserving operations, and bipartite states with a fixed full-rank marginal state $\rho_A$. Following Ref.~\cite{Kiukas2017},
a bipartite state $\rho$ defines a channel $\Lambda_\rho: \mathcal{M}(\mathcal{H}^A) \rightarrow \mathcal{M}(\mathcal{H}^B)$ in the Heisenberg picture as
\begin{equation}\label{Eq:ChoiHeisenberg}
    \rho_A^{\frac{1}{2}} \Lambda_\rho^*(B_{b|y})^T \rho_A^{\frac{1}{2}}:=\tr_B(\rho^{AB}(\id^A \otimes B_{b|y}^B))
\end{equation}
where $\rho_A = \tr_B(\rho^{AB})$ and the transposes is taken in the eigenbasis of the state $\rho_A$. We note that the left-hand-side of Eq.~(\ref{Eq:ChoiHeisenberg}) is simply an operation $\mathcal{S}^*_\rho: \mathcal{M}(\mathcal{H}^{B}) \rightarrow \mathcal{M}(\mathcal{H}^A)$ acting on Bob's POVMs. We define this formally as
\begin{equation} \label{eq:Steering}
    \mathcal{S}^*_\rho(B_{b|y}) := \rho_A^{\frac{1}{2}} \Lambda_\rho^*(B_{b|y})^T \rho_A^{\frac{1}{2}}
    =\sigma_{b|y}^A .
\end{equation}

For later convenience, we note that the Schrödinger pictures of the above maps are given by
\begin{equation} \label{eq:channel_state}
    \Lambda_\rho(\tau) = \tr_A\left( \rho_2^{AB} \left( \left(\rho_A^{-\frac{1}{2}} \tau^T \rho_A^{-\frac{1}{2}}\right)^A \otimes \id^B \right) \right) \:
\end{equation}
and
\begin{align} \label{eq:operation_state}
    \mathcal{S}_\rho(\tau) 
    = \tr_A\left( \rho^{AB} (\tau^{A} \otimes \id^B) \right)
    = \Lambda_\rho(\rho_A^{\frac{1}{2}}\tau^T\rho_A^{\frac{1}{2}} ) \: .
\end{align}
In the following, we apply this formalism to more involved scenarios, starting with the bilocality scenario and going from there to more complex line networks.

\begin{figure}
    \centering
    \includegraphics[width= 0.5\linewidth, clip]{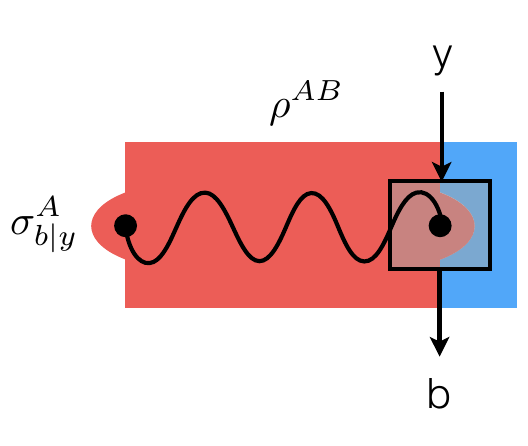}
    \caption{\textbf{Bipartite Quantum Steering.} The dots represent Alice's (on the left) and Bob's (on the right) shares of an entangled state $\rho^{AB}$. The shared state can be viewed as an operation mapping the untrusted measurement set to a state assemblage $\sigma_{b|y}^A$}.
    \label{fig:twoParty}
\end{figure}

\subsection{Bilocality Scenario}

In the bilocality scenario, cf. Fig.~\ref{fig:Bilocality}, there are two independent bipartite states $\rho_1^{AB}$ and $\rho_2^{B'C}$. On these states, a global measurement $B_b^{BB'}$ is performed on systems $B$ and $B'$ and local measurements are performed on $A$ and $C$.
This scenario provides an example of a quantum network. In this work, we concentrate on networks with trusted endpoints. Hence, the scenario can be viewed as a combination of two bipartite steering scenarios where an untrusted measurement $B_b^{BB'}$ acts globally on two systems. We note that unlike in the standard steering scenario, in networks measurement inputs $y$ are not required to certify entanglement properties.
Therefore, in the following it is assumed that the measurements have no inputs. However, it is straight forward to extend the results to networks with measurement inputs. Changing from shared states into the operation picture via Eq.~(\ref{eq:Steering}), we get the following observation:
\begin{observation}
The bilocality scenario with trusted endpoints can be written in the operations picture as follows    
\begin{align} \label{eq:Bilocality}
    \sigma_b^{AC} &=\tr_{BB'}\left( (\id^A \otimes B_b^{BB'} \otimes \id^C) (\rho^{AB}_1 \otimes \rho^{B'C}_2) \right) \nonumber\\
     &= (\mathcal{S}^*_{\rho_1} \otimes \mathcal{S}^*_{\rho_2})(B_b) \: ,
\end{align}
where $\mathcal{S}^*_{\rho_1}: \mathcal{M}(\mathcal{H}^B) \rightarrow \mathcal{M}(\mathcal{H}^A)$ and $\mathcal{S}^*_{\rho_2}: \mathcal{M}(\mathcal{H}^{B'}) \rightarrow \mathcal{M}(\mathcal{H}^C)$ are operations corresponding to the states $\rho^{AB}_1$ and $\rho^{AB}_2$, respectively. They map the global measurement to a state ensemble, i.e. a state assemblage with no input $y$.
\end{observation}

In what follows, it is crucial that the operations $\mathcal{S}^*_{\rho_1}$ and $\mathcal{S}^*_{\rho_2}$ act locally with respect to the cut $AB:B'C$. This is a manifestation of the assumption of independent sources. The locality of the operations is used extensively throughout the certification proofs in the next sections.

Eq.~(\ref{eq:Bilocality}) shows that the operation formalism allows one to write down the action of the network on the global measurement in a compact manner. In contrast to mapping the global measurement, the focus can be placed on mapping one of the bipartite states $\rho_1^{AB}$. 
For this the measurement $B_b^{BB'}$ and the state $\rho_2^{B'C}$ are required to be written as one combined map from system $B$ to system $C$, cf. right-hand side of Fig.~\ref{fig:Bilocality}.
This map is a concatenation of the operation related to the measurement $B_b^{BB'}$ and operation related to the state $\rho_2^{B'C}$.

First, we note that unlike in the case of mapping measurements, where a Heisenberg operation is used, here we use the Schrödinger picture as we are mapping states. Consider the state $\rho_2^{B'C}$. The Schrödinger picture of the operation $\mathcal{S}_{\rho_2}$ associated with this state is given in Eq.~(\ref{eq:operation_state}). Here, $\hil^A$ and $\hil^B$ need to be relabeled to $\hil^{B'}$ and $\hil^C$, respectively.

Second, consider the action of the measurement $B_b$.
To determine the associated operation, a slightly modified version of the Choi-Jamiolkowski isomorphism is required. For this, we follow Ref.~\cite{Shirokov2008}.
The formalism has a technical requirement that the POVM elements need to have trace at most one.
Therefore, as an intermediate step the operation associated with $\sfrac{B_b}{d^2}$ is determined.
As another technical requirement, the marginal state $\chi$ of the Choi state has to satisfy $\chi \geq \tr_B\left(\sfrac{B_b}{d^2}\right)$, for which $\chi = \sfrac{\id}{d}$ is a valid choice.
According to Ref.~\cite{Shirokov2008}, the operation $\mathcal{B}_{B_{b}}: \mathcal{M}(\mathcal{H}^{B}) \rightarrow \mathcal{M}(\mathcal{H}^{B'})$ associated with $\sfrac{B_b}{d^2}$ can be written as
\begin{equation}
    \mathcal{B}_{B_{b}}(\tau) = d \tr_B \left( \frac{B_b^{BB'}}{d^2} \left((\tau^B)^T \otimes \id^{B'}\right)\right) \: .
\end{equation}
We can now write the map induced by a POVM element in the bilocality scenario more compactly as
\begin{align}
    & \tr_B\left( B_b^{BB'}(\tau^B \otimes \id^{B'})\right)
     = d \mathcal{B}_{B_b}(\tau^T)
     =: \mathcal{B}'_{B_b}(\tau)
\end{align}
where $\mathcal{B}'_{B_{b}}: \mathcal{M}(\mathcal{H}^{B}) \rightarrow \mathcal{M}(\mathcal{H}^{B'})$ is a completely positive map.

Finally, the combined action of the measurement and the shared state $\mathcal{E}_{B_b,\rho_2}: \mathcal{M}(\mathcal{H}^B) \rightarrow \mathcal{M}(\mathcal{H}^C)$ is defined as the concatenation of the individual actions
\begin{equation}  \label{eq:Def_BuildingBlock}
    \mathcal{E}_{B_b,\rho_2} (\tau) = \left(\mathcal{S}_{\rho_2} \circ \mathcal{B'}_{B_b}\right)(\tau) \: .
\end{equation}
This is what we call the \textit{building block} of a network. Notice that for notational convenience, we have taken the factor $d$ into the definition of $\mathcal{B'}_{B_b}$, which may cause the building block to be trace-increasing. However, for our certification proofs it is only crucial that the building blocks act locally. This map is illustrated in Fig.~\ref{fig:BuildingBlock} where the blue block represents $\mathcal{B}'_{B_b}$ and the red block represents $\mathcal{S}_{\rho_2}$ and they are combined to a purple building block $\mathcal{E}_{B_b,\rho_2}$.
The use of a building block in a bilocality network is summarised in the following observation, cf. right-hand side of Fig.~\ref{fig:Bilocality}:

\begin{observation} 
When placing a state $\rho_1^{AB}$ in a bilocality network characterised by a building block $\mathcal{E}_{B_b,\rho_2}$, the state ensemble at the endpoints is given by 
\begin{equation}\label{Eq:StateTransmit}
    \sigma_b^{AC} = \left(\text{id}_A \otimes \mathcal{E}_{B_b,\rho_2}\right) (\rho_1^{AB}) \:.
\end{equation}
\end{observation}

\begin{figure}
    \centering
    \includegraphics[width= 0.9\linewidth, clip]{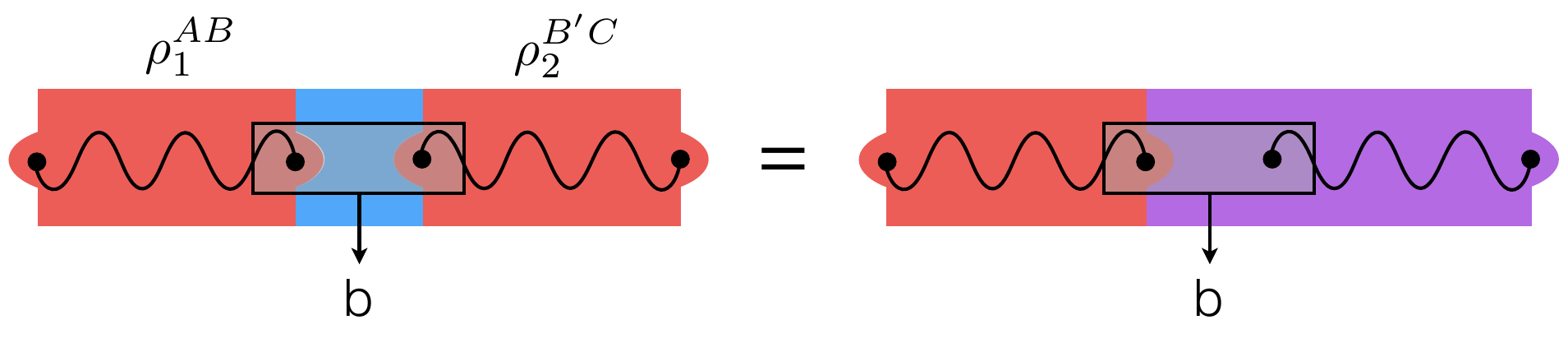}
    \caption{\textbf{Bilocality Scenario.} There are two shared bipartite states and one untrusted measurement. This can be interpreted as two bipartite steering scenarios with a global measurement (left-hand side), as formalised by Eq.~(\ref{eq:Bilocality}). Alternatively, this can be seen as transmitting part of a state through an operation defined by a measurement and a bipartite state (right-hand side), as formalised by Eq.~(\ref{Eq:StateTransmit}).}
    \label{fig:Bilocality}
\end{figure}

\begin{figure}
    \centering
    \includegraphics[width= 0.9\linewidth, clip]{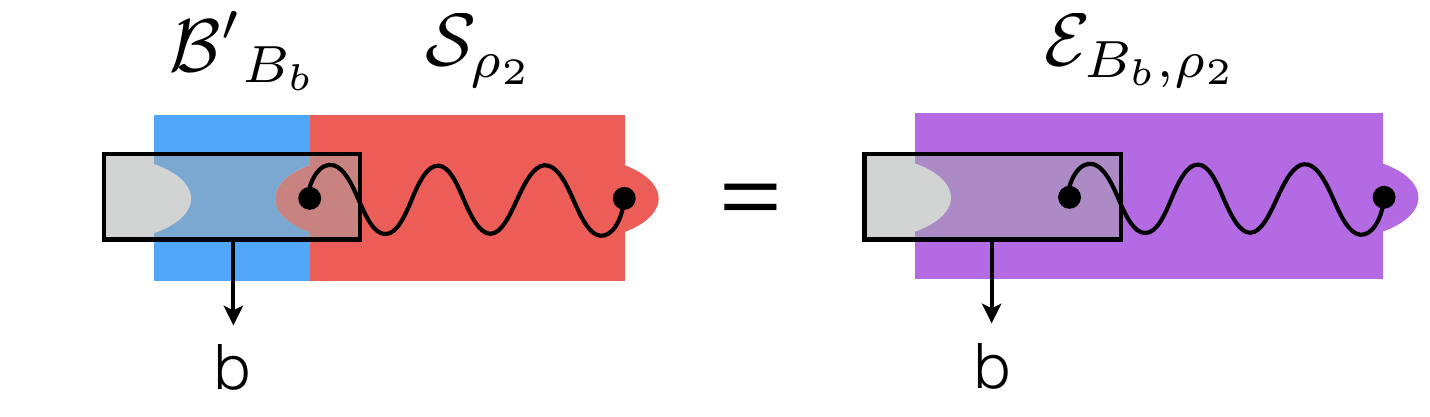}
    \caption{\textbf{Building blocks of line networks.} Aside from a state or measurement of interest a line network can be viewed as a concatenation of a map defined by a measurement and a shared bipartite state.}
    \label{fig:BuildingBlock}
\end{figure}

The introduced building block $\mathcal{E}_{B_b,\rho_2}$ appears in a variety of quantum communication scenarios, e.g. entanglement swapping, quantum teleportation, and entanglement-assisted prepare-and-measure scenarios.
In this work, we focus on using the building blocks when certifying line networks with trusted end points~\footnote{Notice that this is equal to a circular network with one trusted party. Circular networks of this form are simply line networks where the trusted end points are treated as one party.}.

\subsection{Line network}

As more involved networks require the use of various Hilbert spaces, from here on we drop the labels $A,B,C$ etc. of the space when there is no risk of confusion.
In the operations picture, a line network can be decomposed into a bilocality network with a series of measurement-state pairs on either side.
Concretely, say a bipartite state $\rho$ is placed into a line network, cf. top line of Fig.~\ref{fig:LineNetwork}.
Then there are $m$ measurements $B_{b_j}$ and $m$ states $\omega_j$ to its left, each pair forming a building block $\mathcal{E}_{B_{b_j},\omega_j}$. Thus, the total action of all measurements and states to the left of the state are given by the concatenation
\begin{equation}
    \mathcal{A}_{\vec{b}} := \mathcal{E}_{B_{b_m},\omega_m} \circ .. \circ \mathcal{E}_{B_{b_1},\omega_1}
\end{equation}
where $\vec{b} = (b_1,b_2,...,b_m)$.
Note that if there are no measurements and states to the left of the considered state $\rho$ (e.g. see right-hand side of Fig.~\ref{fig:Bilocality}), then $m=0$ and hence, $\mathcal{A} = \text{id}$.
Say in the line network under consideration there are also $n$ measurements $C_{c_j}$ and $n$ states $\xi_j$ on the right-hand side of the considered state $\rho$. These form building blocks $\mathcal{E}_{C_{c_j},\xi_j}$. Then the combined action of all measurements and states on the right-hand side is written as
\begin{equation}
    \tilde{\mathcal{A}}_{\vec{c}} := \mathcal{E}_{C_{c_n},\xi_n} \circ .. \circ \mathcal{E}_{C_{c_1},\xi_1}
\end{equation}
where $\vec{c} = (c_1,c_2,...,c_n)$.
Note that if there are no measurements or states to the right of the considered state $\rho$, $n=0$ and hence, we set $\tilde{\mathcal{A}} = \text{id}$. We summarise this in the following.
\begin{observation}
When inserting a state $\rho$ into a line network charaterised by the concatenated building blocks $\mathcal{A}_{\vec{b}}$ and $\tilde{\mathcal{A}}_{\vec{c}}$ acting locally on $\rho$, the state ensemble at the endpoints reads
\begin{equation} \label{eq:network_state}
    \sigma_{\vec{b},\vec{c}} = (\mathcal{A}_{\vec{b}} \otimes \tilde{\mathcal{A}}_{\vec{c}}) (\rho) \: .
\end{equation}
\end{observation}
Alternatively, the focus can be placed on a measurement $\{M_i\}_i$ in the line network, c.f. bottom line of Fig.~\ref{fig:LineNetwork}.
We can take stock of the above Observation by noting that the states $\omega_0$ and $\xi_0$ neighbouring the measurement act as local operations $\mathcal{S}^*_{\omega_0}$ and $\mathcal{S}^*_{\xi_0}$. These map the POVM elements of $M_i$ to subnormalised states.
Subsequently, the remaining measurements and states of the network act in the same manner as presented previously.
\begin{observation}
When inserting a bipartite measurement $\{M_i\}_i$ into a line network, the state ensemble at the endpoints is given by
\begin{align} \label{eq:network_meas}
    \sigma_{i,\vec{b},\vec{c}} 
    & = (\mathcal{A}_{\vec{b}} \otimes \tilde{\mathcal{A}}_{\vec{c}}) (\mathcal{S}^*_{\omega_0} \otimes \mathcal{S}^*_{\xi_0})(M_i) \nonumber \\
    & = (\mathcal{A}_{\vec{b}} \circ \mathcal{S}^*_{\omega_0} \otimes \tilde{\mathcal{A}}_{\vec{c}} \circ \mathcal{S}^*_{\xi_0}) (M_i) \: .
\end{align} 
\end{observation}

\begin{figure}[t!]
    \centering
    \includegraphics[width= 0.9\columnwidth, clip]{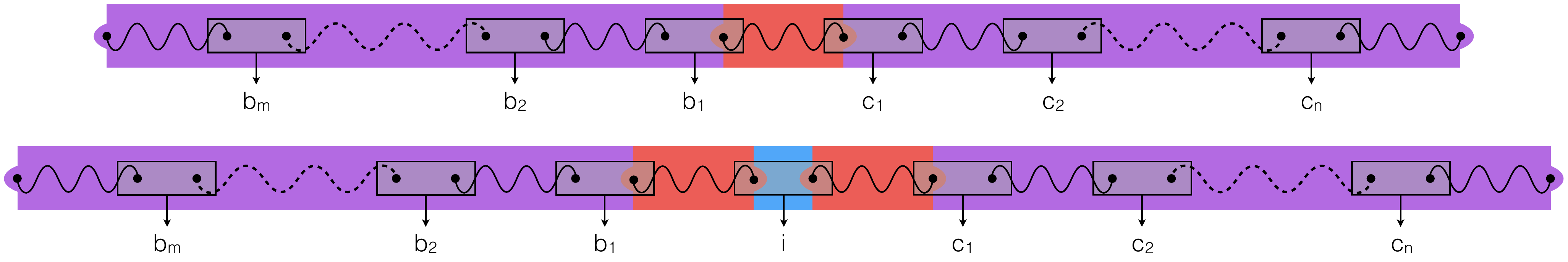}
    \caption{\textbf{Line network.} A line network can be viewed as a state on which a concatenation of building blocks, i.e. pairs of a measurement and a state, acts locally on either side (top line), as formalised in Eq.~(\ref{eq:network_state}).
    Alternatively, the scenario can be viewed as a measurement which is placed in a bilocality network extended by building blocks on both sides (lower line), as formalised in Eq.~(\ref{eq:network_meas}).}
    \label{fig:LineNetwork}
\end{figure}

\section{Ideal line networks}

An optimal building block for certifying properties of states or measurements in a network is one that has minimal noise. For example, a building block consisting of a Bell state and a Bell measurement corresponds to a perfect teleportation scenario. In this case the ensemble at the trusted endpoints will encode all relevant quantum information of the state or measurement we wish to certify. Hence, any quantum property that is not created under local actions, can be certified from the trusted endpoints. This constitutes a semi-device independent witness, as the network itself does not need to be trusted apart from having independent sources, i.e. maps acting locally. It is worth noting, that an idealised certification scenario can even be made device-independent for example in the case of entanglement certification of states \cite{PhysRevLett.121.180503}. 

However, we are interested in certifying states and measurements in imperfect networks. For example, in the above case of perfect Bell states and Bell state measurements, one can introduce white noise to each state with visibility $v$. Thus, the state ensemble has visibility $v^k$ when we have $k$ building blocks, hence, limiting the certifying power of the network. We formalise this idea in the next section to include more general noise models than the visibility.

\section{Robustness and convex weight based entanglement quantifiers}

Utilising the presented framework, it is clear that if a conditional state prepared by the network is entangled, all the states in the network have to be entangled and the measurements have to be entangling.
The framework also suggests that the entanglement of states and measurements in a network can be quantified from the trusted endpoints.
The only requirement appears to be that a chosen quantifier does not increase under local operations.
In the following we will present different entanglement quantifiers which can be lower bound in a line network. Specifically, we will focus on robustness- and weight-based quantifiers in this section, as well as on the Schmidt number as an example of a dimension-based quantifier in the next section.

\subsection{Certification of states}

The robustness of a state $\rho$ with respect to a convex and closed set $\mathcal{F}$, called the free set, corresponds to the amount of noise that is required for the state to enter the set $\mathcal{F}$. Formally, it is defined as
\begin{align}
    \text{R}_\mathcal{F}^\mathcal{N}(\rho) = \min r \geq 0 \text{ s.t. } \rho = (1+r) \tau - r \eta,
\end{align}
where $\tau$ runs over states in the free set $\mathcal{F}$ and $\eta$ runs over another convex and compact subset $\mathcal{N}$ of states, called the noise set. For example, in the case of $\mathcal{F}$ being the set of separable states, the robustness expresses how much noise of a given type can be added until the state becomes separable.

To illustrate the idea of quantifying the robustness in a network, consider a bilocality network and depolarizing noise. More precisely, we take the noise states to have the same marginal on the untrusted side as the original one.
In this case, if a state $\rho$ in the bilocality network has robustness $r$, then the prepared state ensemble $\sigma_b$ takes the form
\begin{align} \label{eq:Depolar}
    \sigma_b
    & = (\text{id} \otimes \mathcal{E}_{B_b,\omega}) (\rho)   
    = \tr(\sigma_b)\left( (1+r) \hat{\tilde{\tau}}_b -r \hat{\tilde{\eta}}_b \right)
\end{align}
where $\hat{\tilde{\tau}}_b \in \mathcal{F}$, $\hat{\tilde{\eta}}_b \in \mathcal{N}$, and a hat on top of an operator denotes that the operator has trace one.
Here we used that $\tr(\sigma_b) =  \tr(\tilde{\eta}_b) =  \tr(\tilde{\tau}_b)$ due to the particular choice of the noise set $\mathcal{N}$, i.e. 
\begin{align}
    \tr(\tilde{\eta}_b)
    & = \tr\left( (B_b^{BB'} \otimes \id^C) (\tr_A(\eta^{AB})\otimes \omega^{B'C}) \right) \nonumber\\
    & = \tr\left( (B_b^{BB'} \otimes \id^C) (\tr_A(\rho^{AB})\otimes \omega^{B'C}) \right) \nonumber\\
    & = \tr(\sigma_b)
\end{align}
where $\eta$ comes from the definition of the robustness.
Thus, also $\tr(\sigma_b) =  \tr(\tilde{\tau}_b)$.
The decomposition given in Eq.~(\ref{eq:Depolar}) has the same form as required by the robustness. 
Thus, the robustness of the prepared state $\sigma_b$ is upper bounded by the robustness $r$ of the state $\rho$.
Conversely, this implies that the robustness of both states in the bilocality network is directly lower bound by the robustness of the induced states at the endpoints. We formalise this in the following Observation.

\begin{observation}
    Assume that each of the cones spanned by the free set $\mathcal{F}$ and the noise set $\mathcal{N}$ are closed under local operations and that the noise is depolarizing. Further, let $\sigma_{\vec{b},\vec{c}}$ denote the state ensemble at the trusted end-points of a line network.
    Then, we have the following inequality:
    \begin{equation}
        \text{R}_\mathcal{F}^\mathcal{N}(\hat{\sigma}_{\vec{b},\vec{c}}) \leq \text{R}_\mathcal{F}^\mathcal{N}(\rho) \: ,
    \end{equation}
    where $\rho$ is a state in the network.
\end{observation}

In contrast to the noise tolerance given by the robustness measure, the convex weight of a state $\rho$ indicates the fraction of resource included in the state:
\begin{align}
    \text{W}_\mathcal{F}(\rho) = \min w \geq 0 \text{ s.t. } \rho = (1-w) \tau + w \eta,
\end{align}
where the minisation is over $\tau \in \mathcal{F}$ and $\eta$ being any state.
For example, when $\mathcal{F}$ is the set of separable states, the quantity $1-\text{W}_\mathcal{F}(\rho)$ expresses how much of the mixture can be written as a separable state.

Based on the previously introduced notation (cf. Eq.~(\ref{eq:network_state})), it is intuitive that both of the presented entanglement quantifiers do not increase when placing a state in a line network.
Concretely, the following observation is obtained.

\begin{observation} \label{obs:RWstates}
Assume that each of the cones spanned by the free set $\mathcal{F}$ and the noise set $\mathcal{N}$ are closed under local operations and denote by $\sigma_{\vec{b},\vec{c}}$ the state ensemble at the trusted end-points of a line network. Then, we have the following inequalities:
\begin{enumerate}
    \item \textit{Robustness of an individual state:}
            \begin{equation}
                \sum_{\vec{b},\vec{c}}\tr(\sigma_{\vec{b},\vec{c}}) \text{R}_\mathcal{F}^\mathcal{N}(\hat{\sigma}_{\vec{b},\vec{c}})
                \leq \text{R}_\mathcal{F}^\mathcal{N}(\rho)
            \end{equation}
            where $\rho$ is a state in the network.
    \item \textit{Convex weight of all states:}
            \begin{equation}
                \sum_{\vec{b},\vec{c}}\tr(\sigma_{\vec{b},\vec{c}}) \text{W}_\mathcal{F}(\hat{\sigma}_{\vec{b},\vec{c}})
                \leq \prod_{j=1}^n \text{W}_\mathcal{F}(\rho_j)
            \end{equation}
            where $\rho_j$ are states in the network.
\end{enumerate}

\end{observation}

\begin{proof}
First, proving statement 1.
Similar to the bilocality scenario, if a state $\rho$ with robustness $r$ is placed in a line network,
the prepared state assemblage $\sigma_{\vec{b},\vec{c}}$ takes the form
\begin{align} \label{eq:assemblage_decomposition}
    \sigma_{\vec{b},\vec{c}} 
    & = (\mathcal{A}_{\vec{b}} \otimes \tilde{\mathcal{A}}_{\vec{c}}) (\rho)    
    = (1+r) \tilde{\tau}_{\vec{b},\vec{c}} -r \tilde{\eta}_{\vec{b},\vec{c}}
\end{align}
where $\tilde{\tau}_{\vec{b},\vec{c}} \in \text{cone}(\mathcal{F})$, $\sum_{\vec{b},\vec{c}} \tr(\tilde{\tau}_{\vec{b},\vec{c}})=1$, $\tilde{\eta}_{\vec{b},\vec{c}} \in \text{cone}(\mathcal{N})$ and $\sum_{\vec{b},\vec{c}} \tr(\tilde{\eta}_{\vec{b},\vec{c}})=1$.
This decomposition resembles the format required by the robustness up to normalisation where $r\tr(\tilde{\eta}_{\vec{b},\vec{c}})$ takes the role of the robustness.
Thus, by dividing by the trace of the state $\sigma_{\vec{b},\vec{c}}$, this yields an upper bound on the robustness R$_\mathcal{F}^\mathcal{N}$ of the normalised prepared state $\hat{\sigma}_{\vec{b},\vec{c}}$
\begin{align} \label{eq:RobustnessBound}
   \text{R}_\mathcal{F}^\mathcal{N}(\hat{\sigma}_{\vec{b},\vec{c}}) \leq \frac{\tr(\tilde{\eta}_{\vec{b},\vec{c}})}{\tr(\sigma_{\vec{b},\vec{c}})} r \: .
\end{align}
Conversely, this can be viewed as a lower bound on the robustness $r$ of the state $\rho$.
Note that in general $\tr(\tilde{\eta}_{\vec{b},\vec{c}})$ is unknown.
However, summing over the measurement outcomes leads to the robustness bound presented previously.
Note that the presented arguments hold for every state in the line network.
Therefore, the weighted average of the robustness lower bounds the robustness of every state in the line network individually.

Secondly, consider the second statement.
Now say a state $\rho$ has convex weight $w$.
Then the state $\sigma_{\vec{b},\vec{c}}$ prepared by the line network takes the form presented in Eq.~(\ref{eq:assemblage_decomposition}) where $r$ is repalced by $-w$.
Thus, in this case the decomposition of the assemblage resembles the decomposition required by the convex weight up to normalisation where $w\tr(\tilde{\eta}_{\vec{b},\vec{c}})$ plays a similar role as the convex weight.
These arguments hold for all states in the network.
Intuitively, it is clear that only the part of the mixture of each state which is not in the free set contributes to the part of the prepared state which is not in the free set and thus, up to normalisation, the convex weight of the state assemblage is upper bound by the product of the convex weights $w_j$ of all states $\rho_j$ in the network (see Appendix~\ref{app:allStates} for the derivation).
So when taking the normalisation into account, an equation similar to Eq.~(\ref{eq:RobustnessBound}) is derived where $r$ is replaced by $\prod_{j=1}^n w_j$ and the robustness of the state on the trusted side by its convex weight.
By summing over the measurement outcomes, the previously presented average convex weight bound is retrieved.
\end{proof}

Typically, the states and measurements of the network are assumed to be independent.
However, say the untrusted components of the network have access to some (classical) shared randomness $\lambda$, which is unknown to the trusted end points.
In this case the bounds from Observation~\ref{obs:RWstates} can be rephrased in the following way:
\begin{corollary} \label{cor:RW_state}
    Keeping the assumptions from Observation~\ref{obs:RWstates} but allowing  for the states and measurements of the line network to depend on some shared randomness $\lambda$ with probability distribution $p(\lambda)$, then the previously presented bounds take the form
    \pagebreak
    \begin{enumerate}
        \item \textit{Robustness of an individual state with shared randomness:}
            \begin{equation}
                \sum_{\vec{b},\vec{c}}\tr(\sigma_{\vec{b},\vec{c}}) \text{R}_\mathcal{F}^\mathcal{N}(\hat{\sigma}_{\vec{b},\vec{c}})
                \leq \sum_\lambda p(\lambda) \text{R}_\mathcal{F}^\mathcal{N}(\rho^\lambda)
            \end{equation}
            where $\rho^\lambda$ is a state in the network depending on $\lambda$.
        \item \textit{Convex weight of all states with shared randomness:}
            \begin{equation}
                \sum_{\vec{b},\vec{c}}\tr(\sigma_{\vec{b},\vec{c}}) \text{W}_\mathcal{F}(\hat{\sigma}_{\vec{b},\vec{c}})
                \leq \sum_\lambda p(\lambda) \prod_{j=1}^n \text{W}_\mathcal{F}(\rho_j^\lambda)
            \end{equation}
            where $\rho_j^\lambda$ are states in the network depending on $\lambda$.
    \end{enumerate}
\end{corollary}
\begin{proof}
    In the case of shared randomness, the prepared state ensemble takes the form
    \begin{align}
        \sigma_{\vec{b},\vec{c}}
        = \sum_\lambda p(\lambda) \left( \mathcal{A}_{\vec{b}}^\lambda \otimes \tilde{\mathcal{A}}_{\vec{c}}^\lambda\right)(\rho^\lambda)
        =: \sum_\lambda p(\lambda) \sigma_{\vec{b},\vec{c}}^\lambda
    \end{align}
    where $\rho^\lambda$ is a state in the network which depends on the shared randomness $\lambda$.
    Therefore, by Observation~\ref{obs:RWstates} the following holds for the robustness
    \begin{align}
        \sum_{\vec{b},\vec{c}}\tr(\sigma_{\vec{b},\vec{c}}) \text{R}_\mathcal{F}^\mathcal{N}(\hat{\sigma}_{\vec{b},\vec{c}})
        & \leq \sum_{\vec{b},\vec{c}} \sum_\lambda p(\lambda) \tr(\sigma_{\vec{b},\vec{c}}^\lambda) \text{R}_\mathcal{F}^\mathcal{N}(\hat{\sigma}_{\vec{b},\vec{c}}^\lambda) \nonumber\\
        & \leq \sum_\lambda p(\lambda) \text{R}_\mathcal{F}^\mathcal{N}(\rho^\lambda)
    \end{align}
    and similarly for the convex weight
    \begin{align}
        \sum_{\vec{b},\vec{c}}\tr(\sigma_{\vec{b},\vec{c}}) \text{W}_\mathcal{F}(\hat{\sigma}_{\vec{b},\vec{c}})
        & \leq \sum_{\vec{b},\vec{c}} \sum_\lambda p(\lambda) \tr(\sigma_{\vec{b},\vec{c}}^\lambda) \text{W}_\mathcal{F}(\hat{\sigma}_{\vec{b},\vec{c}}^\lambda) \nonumber\\
        & \leq \sum_\lambda p(\lambda) \prod_{j=1}^n \text{W}_\mathcal{F}(\rho_j^\lambda) \: .
    \end{align}
\end{proof}

\subsection{Certification of measurements}

Similar to the certification of the robustness measure and the convex weight for states, lower bounds on the analogues quantifiers for measurements can be derived.
The robustness of a POVM $M = \{M_i\}_i$ is defined by
\begin{align}\label{Eq:MeasRob}
    \text{R}_\mathcal{F}^\mathcal{N} (M) = \min r \geq 0 \text{ s.t. } M_i = (1 + r) P_i - r N_i \:\forall i
\end{align}
where $\{P_i\}_i$ runs over POVMs in the free set $\mathcal{F}$ and $\{N_i\}_i$ runs over POVMs in the noise set $\mathcal{N}$.
In this section only the generalised and the free robustness are considered, i.e. the noise set $\mathcal{N}$ is the set of all POVMs or equal to the free set $\mathcal{F}$, respectively.
The convex weight of a POVM $M = \{M_i\}_i$ is defined by
\begin{align}\label{Eq:MeasWei}
    \text{W}_\mathcal{F}(M) = \min w \geq 0 \text{ s.t. } M_i = (1 - w) P_i + w N_i \:\forall i
\end{align}
where $\{P_i\}_i \in \mathcal{F}$ and $\{N_i\}_i$ is a POVM.

In the following the free set $\mathcal{F}$ is chosen to be the set SEP of separable measurements, i.e. measurements of the form
\begin{equation} \label{eq:def_SepPOVM}
    M_i=\sum_a A_a^i\otimes B_a^i
\end{equation}
with positive semi-definite operators $\{A_a^i\}_{a,i}$ and $\{B_a^i\}_{a,i}$. It turns out that this definition is equivalent to a form that is more suitable for networks 
\begin{equation} \label{eq:SEP_POVM}
    \alpha^{\frac{1}{2}} \otimes \beta^{\frac{1}{2}} (\Lambda^*\otimes\Gamma^*)(M_i) \alpha^{\frac{1}{2}} \otimes \beta^{\frac{1}{2}}
    \in \text{cone(SEP)}
\end{equation}
where $\alpha$ and $\beta$ are any state, and $\Lambda$ and $\Gamma$ are any local operation, see Appendix~\ref{app:SEPequivalence}.

When certifying measurement quantifiers in our scenario, the relevant information is encoded in a set of conditional states at the end points of the network.
Thus, the free set $\mathcal{F}$, and in the case of the robustness the noise set $\mathcal{N}$, need to be translated into a set in state space.
Here, $\mathcal{F}=$ SEP.
The corresponding free set $\mathcal{F}'$ for states is thus chosen to be the set SEP of separable states. In the following, the label SEP is used for both sets, as the argument of the corresponding resource measure already indicates whether separable measurements or separable states are referenced.

\begin{observation} \label{obs:RW_meas}
    If the free set $\mathcal{F}$ is the set SEP of separable measurements and hence, the corresponding free set $\mathcal{F}'$ is the set SEP of separable states, then
    \begin{enumerate}
        \item \textit{Robustness of an individual measurement:}
                \begin{equation}
                    \sum_{i,\vec{b},\vec{c}}\tr(\sigma_{i,\vec{b},\vec{c}}) \text{R}_{\text{SEP}}^{\mathcal{N}'}(\hat{\sigma}_{i,\vec{b},\vec{c}})
                    \leq \text{R}_{\text{SEP}}^\mathcal{N}(M)
                \end{equation}
                where $M$ is a measurement in the line network, 
                and $\mathcal{N}$ and $\mathcal{N}'$ are both either the separable set or the whole set (of POVMs or states).
        \item \textit{Convex weight of the whole network:}
                \begin{align}
                    & \sum_{i,\vec{b},\vec{c}}\tr(\sigma_{i,\vec{b},\vec{c}}) \text{W}_{\text{SEP}}(\hat{\sigma}_{i,\vec{b},\vec{c}}) \nonumber\\
                    & \leq \left(\prod_{j=1}^{n+1} \text{W}_{\text{SEP}}(M^j)\right)\left(\prod_{j=1}^m\text{W}_{\text{SEP}}(\rho_j)\right)
                \end{align}
                where $M^j$ are measurements and $\rho_j$ are states in the line network.
    \end{enumerate}
\end{observation}

\begin{proof}
Before proving the presented results, let us consider separable measurements.
By definition, a separable POVM remains separable when sent through local channels.
Further, from Eq.~(\ref{eq:SEP_POVM}) it is clear that in a bilocality network such a POVM induces a separable state.
Therefore, since a line network can be viewed as a bilocality network with additional local operations, any separable measurement prepares a separable state ensemble.

Let us first derive the bound on the robustness.
Consider a POVM $M = \{M_i\}_i$ with robustness $r$.
If this measurement is placed in a line network, the prepared states take the form
\begin{align}
    \sigma_{i,\vec{b},\vec{c}} 
    & = (\mathcal{A}_{\vec{b}} \otimes \tilde{\mathcal{A}}_{\vec{c}})(\mathcal{S}^*_{\omega_0} \otimes \mathcal{S}^*_{\xi_0})(M_i) \nonumber\\
    & = (\mathcal{A}_{\vec{b}} \otimes \tilde{\mathcal{A}}_{\vec{c}})((1+r)\tau_i -r \eta_i) \label{eq:assemblage_POVM}
\end{align}
where $\tau_i = (\mathcal{S}^*_{\omega_0} \otimes \mathcal{S}^*_{\xi_0})(P_i) \in \text{cone(SEP)}$ and $\eta_i = (\mathcal{S}^*_{\omega_0} \otimes \mathcal{S}^*_{\xi_0})(N_i) \in \text{cone}(\mathcal{N}')$.

Hence, the state ensemble $\sigma_{i,\vec{b},\vec{c}}$ can be seen as a state ensemble prepared by a line network which is probabilistically prepared with a state from the ensemble $\{\pi_i\}_i := \{(\mathcal{S}^*_{\omega_0} \otimes \mathcal{S^*_{\xi_0}})(M_i)\}_i$.
The states $\pi_i$ in the ensemble have at most robustness $r$ when evaluated with respect to the separable set SEP and the noise set $\mathcal{N}'$.
Thus, based on Eq.~(\ref{eq:RobustnessBound}) the following holds for each state $\pi_i$ individually
\begin{equation}
    \sum_{\vec{b},\vec{c}}\tr(\sigma_{i,\vec{b},\vec{c}}) \text{R}_{\text{SEP}}^{\mathcal{N}'}(\hat{\sigma}_{i,\vec{b},\vec{c}})
    \leq \tr(\eta_i) \text{R}_{\text{SEP}}^{\mathcal{N}'}(\pi_i)
    \leq \tr(\eta_i) r \: .
\end{equation}
By summing over the remaining measurement outcome $i$, the presented bound on the robustness is derived.

In order to derive the bound on the convex weight, consider a measurement $M$ with convex weight $w$.
Then the prepared states $\sigma_{i,\vec{b},\vec{c}}$ take a similar form as for the robustness, cf. Eq.~(\ref{eq:assemblage_POVM}) where $r$ is replaced by $-w$.
Similar to the case of the robustness, this means that the network behaves like a network which was prepared with a state ensemble and hence, the bound presented for the states holds.
Furthermore, when considering the states and measurements simultaneously, it is intuitive that the bound extends to all states and all measurements (c.f. detailed derivation in Appendix~\ref{app:allPOVM&States}).
\end{proof}

Analogously to the case of states, the above Observation can be extended to allow for shared randomness. In this case the following holds:
\begin{corollary}
    While keeping the assumptions from Observation~\ref{obs:RW_meas} but allowing for the states and the measurements of the line network to depend on some shared randomness $\lambda$ with probability distribution $p(\lambda)$, the previously presented bounds take the form
    \begin{enumerate}
        \item \textit{Robustness of an individual measurement with shared randomness:}
            \begin{equation}
                \sum_{i,\vec{b},\vec{c}}\tr(\sigma_{i,\vec{b},\vec{c}}) \text{R}_{\text{SEP}}^\mathcal{N'}(\hat{\sigma}_{i,\vec{b},\vec{c}})
                \leq \sum_\lambda p(\lambda) \text{R}_{\text{SEP}}^\mathcal{N}(M^\lambda)
            \end{equation}
            where $M^\lambda$ is a measurement in the network depending on $\lambda$, and $\mathcal{N}$ and $\mathcal{N}'$ are either the separable set or the set of all POVMs and all states, respectively.
        \pagebreak
        \item \textit{Convex weight of the whole network with shared randomness:}
            \begin{align}
                & \sum_{i,\vec{b},\vec{c}}\tr(\sigma_{i,\vec{b},\vec{c}}) \text{W}_{\text{SEP}}(\hat{\sigma}_{i,\vec{b},\vec{c}}) \nonumber\\
                & \leq \sum_\lambda p(\lambda) \left(\prod_{j=1}^n \text{W}_{\text{SEP}}(M_j^\lambda)\right)\left(\prod_{j=1}^m\text{W}_{\text{SEP}}(\rho_j^\lambda)\right)
            \end{align}
            where $M_j^\lambda$ are measurements and $\rho_j^\lambda$ are states in the network depending on $\lambda$.
    \end{enumerate}
\end{corollary}
The derivation follows the same arguments as Corollary~\ref{cor:RW_state}.

\section{Dimension based entanglement quantifiers}

Previously, the certification of robustness- and convex weight-based quantifiers was presented.
In the following, the focus will be placed on dimension-based entanglement quantifiers.
More specifically, we will propose a definition for the Schmidt number of bipartite POVMs as well as take advantage of the operations framework to certify the Schmidt number of states and measurements in a network.

\subsection{Schmidt number of states}

The standard dimension-based entanglement quantifier for bipartite states is the Schmidt number \cite{Terhal_2000}.
The Schmidt number SN of a bipartite state $\rho = \sum_\lambda p(\lambda) \ket{\psi_\lambda}\bra{\psi_\lambda}$ indicates the amount of levels one must be able to entangle in order to create the state as a convex mixture:
\begin{equation}\label{Eq:Snumberstate}
    \text{SN}(\rho) = \min \max_{\lambda} \text{SR}(\ket{\psi_\lambda}),
\end{equation}
where the minimisation is taken over all possible pure-state decompositions of $\rho$. The Schmidt rank SR of a pure bipartite state $\ket{\psi_\lambda}$ is the number of terms in its Schmidt decomposition or, equivalently, the matrix rank of its reduced state.

\begin{observation} \label{obs:SN_state}
    The Schmidt number of a state $\rho$ in a line network is lower bound by the Schmidt number of a state $\sigma_{\vec{b},\vec{c}}$ in the state assemblage
    \begin{equation}
        \text{SN}(\sigma_{\vec{b},\vec{c}}) \leq \text{SN}(\rho) \: \forall \vec{b},\vec{c}.
    \end{equation}
\end{observation}

\begin{proof}
Consider a decomposition of the state $\rho$ giving the optimal solution for Eq.~(\ref{Eq:Snumberstate}).
Then the Schmidt ranks, i.e. the local ranks, of the pure states from the decomposition are upper bounded by the Schmidt number.
Local operations can turn the pure states into (sub-normalised) mixed ones, but can not increase their Schmidt number.
Therefore, if local operations are applied to a bipartite state $\rho$, then the Schmidt number of the prepared state is upper bounded by the Schmidt number of the initial state.
This implies that the Schmidt number of each prepared state lower bounds the Schmidt number of every state in the line network.
\end{proof}

Note that if the line network is ideal except for one state, the actual Schmidt number of that particular state can be determined in a SDI manner, even if it is arbitrarily close to a state with lower Schmidt number.

\subsection{Schmidt number of measurements}

The Schmidt number of a bipartite state is a dimension-based extension of the notion of separability. We take a similar approach to defining the Schmidt number of a bipartite measurement.
From the standard definition of separability for measurements, cf. Eq.~(\ref{eq:def_SepPOVM}), it is not straight-forward to define a dimension-based quantifier. Although we concentrate on finite-dimensional Hilbert spaces, one should note that for a meaningful quantifier, the extendability to the infinite-dimensional case is desirable. More precisely, POVM elements are bounded operators and the space of bounded operators is not equipped with an inner product. Hence, there is no clear sense of direction, orthogonality or partial trace, which is required by the standard Schmidt decomposition of pure states (and by extension mixed states). This issue does not appear with trace-class operators, which form an ideal of bounded operators. We can, hence, avoid the issue by renormalising bounded operators via sandwiching them with a square root of some quantum state. This step is strictly speaking not necessary in the finite-dimensional scenario, but we include it into the definition for completeness. Also, we aim to have a quantifier that is not increased by local operations.
Hence, we choose to use the equivalent notion of separability as presented in Eq.~(\ref{eq:SEP_POVM}) as the basis of our definition.

\begin{definition}
Let $M$ be a POVM. We say that $M$ has Schmidt number $k$ if, under local operations and sandwiching with square-roots of local states, its elements map to the cone of states with at most Schmidt number $k$ and at least one element can be mapped to a state with Schmidt number $k$, up to normalisation. More formally, we have
\begin{align} \label{eq:SN_POVM_b}
     \text{SN}(M) = \max_{\alpha,\beta,i,\Lambda,\Gamma} \text{SN}\left( \frac{1}{N}\alpha^{\frac{1}{2}} \otimes \beta^{\frac{1}{2}} (\Lambda^*\otimes\Gamma^*)(M_i) \alpha^{\frac{1}{2}} \otimes \beta^{\frac{1}{2}} \right)
\end{align}
where $N = \tr\left( \alpha^{\frac{1}{2}} \otimes \beta^{\frac{1}{2}} (\Lambda^*\otimes\Gamma^*)(M_i) \alpha^{\frac{1}{2}} \otimes \beta^{\frac{1}{2}} \right)$.
\end{definition}

It turns out that the definition of the Schmidt number for measurements can be brought into a form that is more fitting for networks. We state this formally in the following Observation.

\begin{observation}
The Schmidt number of a POVM $M$ can be expressed as
\begin{align} \label{eq:SN_POVM}
     \text{SN}(M) = \max_{i,\rho} \text{SN}\left( \frac{1}{N} \tr_{BB'}\left( \left(\id^{AC} \otimes M_i^{BB'}\right) \rho^{AB:B'C}  \right) \right)
\end{align}
where $N = \tr\left( \left(\id^{AC} \otimes M_i^{BB'}\right) \rho^{AB:B'C}  \right)$ and $\rho^{AB:B'C}$ is a state which is separable in the indicated cut.
\end{observation}

\begin{proof}
    One can use the generalised Choi isomorphism to show that any feasible point of each optimisation problem gives a feasible point of the other. First, any feasible point of the optimisation in Eq.~(\ref{eq:SN_POVM_b}) can be mapped to a feasible point of Eq.~(\ref{eq:SN_POVM}) by the use of Choi isomorphism for the operations present in Eq.~(\ref{eq:SN_POVM_b}). Conversely, by noting that product states in the cut $AB:B'C$ are optimal for Eq.~(\ref{eq:SN_POVM}), one can use the inverse of the generalised Choi isomorphism on these states to get a feasible point of Eq.~(\ref{eq:SN_POVM_b}).
\end{proof}

It is worth noting that one could also consider the Schmidt number of a post-measurement state as the definition of Schmidt number for POVMs when restricting to the Lüders rule. However, this would lead to considerations on localisability of the state update \cite{Beckman_2001,PhysRevA.49.4331,Fewster,Polo_G_mez_2022}. Whereas localisability leads to a complementary quantifier for implementability of bipartite measurements \cite{Pauwels_2025}, the above Observation shows that our definition avoids the localisability question. This is due to the fact that the state ensemble at the end points is independent of the used state update by the middle party. This fact is easily checked using properties of the partial trace and the fact that all state updates compatible with a given POVM consist of application of the Lüders update followed by a channel depending on the outcome \cite{Pellonpaa_2013}.

We further note that tracing out the system on which the POVM acts is desirable for Schmidt number, because keeping the additional system may increase the Schmidt number.
For example, if the measurement is a qubit Bell state measurement and the state $\rho^{AB:B'C}$ is the tensor product of two qubit Bell states.
Intuitively, the Schmidt number of the measurement should not be higher than the local dimension.
However, the total post-measurement state has Schmidt number $k=4$.
But after tracing out the measured systems the Schmidt number is indeed two.

\begin{observation}
    The Schmidt number of a measurement $M$ in a line network is lower bound by the Schmidt number of any state $\sigma_{i,\vec{b},\vec{c}}$ in the state assemblage
    \begin{equation}
        \text{SN}(\sigma_{i,\vec{b},\vec{c}}) \leq \text{SN}(M) \: \forall i,\vec{b},\vec{c}
    \end{equation}
\end{observation}

\begin{proof}
From the definition of the Schmidt number of measurements, cf. Eq.~(\ref{eq:SN_POVM}), it is apparent that the Schmidt number of the state prepared by a POVM $M$ with Schmidt number $k$ has at most Schmidt number $k$.
Since a line network only acts as local operations, the Schmidt number of the final state cannot be larger than the Schmidt number of the measurement.
\end{proof}

\section{High-dimensional one-way steering activation}

It is known that there exist states that are one-way steerable, but nevertheless can lead to network steering in both directions \cite{Jones_2021}. In other words, when only a single copy of the state is available, Bob's steering efforts could be explained by a separable state, but when two copies are available, the state changes cannot be explained using only separable resources. Here we show that a similar effect can be reached for high-dimensional steering, i.e. there are states where the steering efforts from a single Bob to Alice can be explained classically, but in a bilocality network the explanation requires high Schmidt number states and measurements.

\begin{observation} \label{obs:Activation}
    Consider the isotropic state with visibility $p$ subject to loss:
    \begin{equation} \label{eq:1waySteer}
	   \rho_{p,q} = q p \ket{\psi^+_d}\bra{\psi^+_d} + q(1-p) \frac{\id_d \otimes \id_d}{d^2} + (1-q) \frac{\id_d}{d}\otimes \ket{\varnothing}\bra{\varnothing} \: .
    \end{equation}
    This state is known to be unsteerable from Bob to Alice for $q \leq (1-p)^{d-1}$~\cite{Sekatski2023}. 
    However, when two copies of the state are put into a bilocality network, one can reach a Schmidt number $k$ ensemble for $p^2 > \frac{d(k-1)-1}{d^2-1}$ and $q >0$.
\end{observation}

\begin{proof}
    Consider a bilocality network where the untrusted measurement has POVM element $B_{b=1}=\ket{\psi^+_d}\bra{\psi^+_d}$.
    Place two copies of the state given in Eq.~(\ref{eq:1waySteer}) into this network such that each state is unsteerable from the untrusted party to the trusted party.
    Then the network prepares the state 
    \begin{align}
    \sigma_{b=1}^{AC}
    & = \left( \mathcal{S}_{\rho_{p,q}}^* \otimes \mathcal{S}_{\rho_{p,q}}^* \right) (B_{b=1}) \nonumber\\
    & = \frac{q^2}{d^2} \left(p^2 \ket{\psi^+_d}\bra{\psi^+_d}^{AC} + (1-p^2) \frac{\id^A \otimes \id^C}{d^2}\right) \: .
    \end{align}
    Note that if $q=0$, the state has trace zero. Thus, for the following step it is assumed that $q>0$.
    Next, we investigate under which condition the prepared state has Schmidt number $k$.
    In order to certify the Schmidt number of a state, Schmidt number witnesses are utilised.
    Here, the witness $ W_k = \id - \frac{d}{k-1} \ket{\psi^+_d}\bra{\psi^+_d}$~\cite{Sanpera2001} is employed.
    This witness is constructed such that $\tr(W_k\rho)<0$ implies that $\rho$ has at least Schmidt number $k$.
    Applying the Schmidt number witness $W_k$ to the state prepared by the bilocality network yields
    \begin{align}
        \tr\left(W_k^{AC} \hat{\sigma}_{b=1}^{AC} \right)
        = 1- \frac{1}{d(k-1)}\left(1+ (d^2-1) p^2\right) \: .
    \end{align}
    Therefore, if $q >0 $ and $p^2 > \frac{d(k-1)-1}{d^2-1}$, then the prepared state has at least Schmidt number $k$.
\end{proof}

For a given dimension $d$ and Schmidt number $k$, the observation yields a specific bound on the visibility $p$ and hence, on the losses $q$.
The larger the Schmidt number $k$, which one wants to certify, the higher the visibility $p$ has to be. However, there always exists a valid choice for $p$ which satisfies the inequality.
In turn the losses have to increase, i.e. $q$ has to decrease, in order to satisfy the imposed constraint to ensure that the state is unsteerable from Bob to Alice.
Due to the exponent $(d-1)$, $q$ tends to zero very quickly. However, the upper bound does not reach zero as long as $p$ is strictly smaller than one.

Unsteerable states correspond to incompatibility breaking channels, i.e. channels which cannot be certified to be entanglement preserving in a scenario where the receiver's measurement is not trusted, cf. Ref.~\cite{Heinosaari_2015,Kiukas2017,Engineer2024}.
Therefore, when choosing $k=2$, the presented example is also an example of two incompatability breaking channels being activated in a bilocality network as they demonstrate entanglement preservation.
Furthermore, if the Schmidt number $k$ is chosen to be greater than two, the channels allow for the preservation of $k$-dimensional entanglement.
Such channels are called $k$-partially entanglement breaking channels \cite{chruscinski2005partiallyentanglementbreakingchannels}. The SDI aspect and certification of partially entanglement breaking channels has formerly been discussed in Refs.~\cite{PhysRevA.107.052425,Engineer2024}
The example presented in Observation~\ref{obs:Activation} implies the following.
\begin{corollary}
    Some incompatibility breaking channels can be SDI certified to be $k$-partially entanglement breaking when two copies are put in parallel to a line network.
\end{corollary}

\section{Conclusion and Outlook}

We have applied the inverse Choi-Jamiolkowski isomorphism to quantum networks with independent bipartite sources and bipartite measurements. The resulting operation-based formalism brings line networks into prepare-and-measure-type correlation scenarios. This allows us to quantify entanglement properties of each state and each measurement between two trusted nodes. As examples, we have demonstrated the quantification techniques for typical convex geometric measures of entanglement: robustness and convex weight.

Moreover, we have defined the concept of the Schmidt number for bipartite measurements. We have shown that in a line network it can be SDI certified, together with the Schmidt number of states. We believe this concept may be of interest when benchmarking some of the currently technologically challenging measurement devices.

We have further used our framework to show an activation result concerning one-way high-dimensional steering. Namely, we have used recently reported high-dimensional states that allow SDI Schmidt number detection in only one direction using the steering scenario \cite{Sekatski2023}. We have demonstrated how such states can nevertheless allow SDI Schmidt number certification in both directions when put into a network. 

For future works, there are various possible directions. First, in practical applications, measurements are often binarised. In our work no assumptions are made about the number of measurement outcomes.
However, if only binary measurements are implemented in a network, the derived bounds can possibly be refined such that the considered quantifier only has to be evaluated for one of the measurement outcomes.
This may reduce the experimental demands significantly when implementing, e.g., a long line network.

Second, from the presented notion of the Schmidt number for measurements, the question arises how to define the Schmidt number for operations. 
This is expected to be interesting, but also challenging due to subtleties concerning the state update after an operation, e.g. regarding the localisability \cite{Beckman_2001,PhysRevA.49.4331,Fewster,Polo_G_mez_2022}. However, note that the notions introduced here do not make any statement about the post-measurement state of the measurement effects. Hence, our quantifier is complementary to the implementation cost of bipartite operations, that has recently regained interest \cite{Pauwels_2025}.

Lastly, in the finite-dimensional case the semi-device independent assumption of bipartite steering translates into the question whether an assemblage is preparable with a separable state \cite{Kogias_2015,Moroder_2016}. However, this equivalence is known to only hold in the finite-dimensional case \cite{jokinen2023compressingcontinuousvariablequantum}. Also, in the infinite-dimensional and possibly continuous-outcome case it may be interesting to investigate whether the definition of Schmidt number of bipartite measurements can be extended in a meaningful manner, as was done for bipartite states in \cite{shirokov2011schmidtnumberpartiallyentanglement}. Hence, further investigation is required if one wishes to extend our results and concepts to the infinite-dimensional case.

\section{Acknowledgments}
We thank Pauli Jokinen and Erkka Haapasalo for pointing out references regarding the generalisation of the Choi-Jamiolkowski isomorphism to operations.
S.E. and R.U. are grateful for the financial support from the Swiss National Science Foundation (Ambizione PZ00P2- 202179) and the Wallenberg Initiative on Networks and Quantum Information (WINQ).

\appendix
\onecolumngrid

\section{Derivation of the bound on the convex weight for all states} \label{app:allStates}

In the following the bound on the convex weight for all states in a network is derived.
Say a line network is made up of measurements $B_{b_j}^j$ and states $\rho_j$ which, according to the definition of the convex weight, can each be decomposed into
\begin{equation}
    \rho_j = (1-w_j) \tau_j + w_j \eta_j,
\end{equation}
where $\tau_j$ is in the free set $\mathcal{F}$.
First, consider the subnormalised state $\sigma_{b_1}$ prepared the by the first two states
\begin{align}
    \sigma_{b_1} 
    & = (\text{id} \otimes \mathcal{E}_{B^1_{b_1},\rho_2})(\rho_1) \label{eq:sig_b1}\\
    & = w_1 (\text{id} \otimes \mathcal{E}_{B^1_{b_1},\rho_2})(\eta_1) + (1-w_1) (\text{id} \otimes \mathcal{E}_{B^1_{b_1},\rho_2})(\tau_1)\\ 
    & = w_1 w_2(\text{id} \otimes \mathcal{E}_{B^1_{b_1},\eta_2})(\eta_1) + w_1(1-w_2) (\text{id} \otimes \mathcal{E}_{B^1_{b_1},\tau_2})(\eta_1) + (1-w_1) (\text{id} \otimes \mathcal{E}_{B^1_{b_1},\rho_2})(\tau_1) \\
    & =: (w_1 w_2) \eta'_{b_1} + (\tr(\sigma_{b_1})-w_1w_2 \tr(\eta'_{b_1})) \hat{\tau}'_{b_1},
\end{align}
where $\sum_{b_1} \tr(\eta'_{b_1}) = 1$ and $\hat{\tau}'_{b_1} \in \mathcal{F}$.
Then consider the subnormalised state $\sigma_{b_1,b_2}$ prepared by the first three states
\begin{align}
    \sigma_{b_1,b_2} 
    = (\text{id} \otimes \mathcal{E}_{B^2_{b_2},\rho_3} \circ \mathcal{E}_{B^1_{b_1},\rho_2}) (\rho_1)
    = (\text{id} \otimes \mathcal{E}_{B^2_{b_2},\rho_3}) (\sigma_{b_1}) \: .
\end{align}
Notice that up to normalisation this has the same form as Eq.~(\ref{eq:sig_b1}) and thus,
\begin{align}
    \sigma_{b_1,b_2} 
    =:  (w_1 w_2 w_3) \eta''_{b_1,b_2} + (\tr(\sigma_{b_1,b_2})-w_1w_2w_3 \tr(\eta''_{b_1,b_2})) \hat{\tau}''_{b_1,b_2},
\end{align}
where $\sum_{b_1,b_2} \tr(\eta''_{b_1,b_2}) = 1$ and $\hat{\tau}''_{b_1,b_2} \in \mathcal{F}$.
This procedure is applicable to all states in the network.
Thus, the subnormalised state $\sigma_{\vec{b}}$ prepared by the whole network takes the form
\begin{align}
    \sigma_{\vec{b}} 
    & = \left(\prod_j w_j\right) \tilde{\eta}_{\vec{b}} + \left(\tr(\sigma_{\vec{b}})-\prod_j w_j \tr(\tilde{\eta}_{\vec{b}})\right) \hat{\tilde{\tau}}_{\vec{b}},
\end{align}
where $\sum_{\vec{b}} \tr(\tilde{\eta}_{\vec{b}}) = 1$ and $\hat{\tilde{\tau}}_{\vec{b}} \in \mathcal{F}$.
The decomposition resembles the decomposition required by the convex weight optimisation up to normalisation where the weight is $\prod_j w_j\tr(\tilde{\eta}_{\vec{b}})$.
Therefore, by diving by the norm of the prepared state $\sigma_{\vec{b}}$, the convex weight of the normalised state $\hat{\sigma}_{\vec{b}}$ is bound by 
\begin{equation}
    \text{W}_\mathcal{F}(\hat{\sigma}_{\vec{b}}) \leq \frac{\tr(\tilde{\eta}_{\vec{b}})}{\tr(\sigma_{\vec{b}})} \left( \prod_j w_j\right) \: .
\end{equation}
By rearranging and summing over the measurement outcomes $\vec{b}$, the bound on the convex weight for all states is derived.
Note that this bound trivially also holds for any subset of states in the network.

\section{Equivalence of definitions of separability} \label{app:SEPequivalence}

In the following, the equivalence between the standard definition of measurement separability, c.f. Eq.~(\ref{eq:def_SepPOVM}), and the expression presented in Eq.~(\ref{eq:SEP_POVM}) will be proven.

First, assume that the standard notion is satisfied.
Then, by substitution
\begin{align}
    \alpha^{\frac{1}{2}} \otimes \beta^{\frac{1}{2}}(\Lambda^*\otimes\Gamma^*) (M_i) \alpha^{\frac{1}{2}} \otimes \beta^{\frac{1}{2}}
    & = \sum_a \alpha^{\frac{1}{2}} \Lambda^*(A_a^i) \alpha^{\frac{1}{2}} \otimes \beta^{\frac{1}{2}} \Gamma^*(B_a^i) \beta^{\frac{1}{2}} 
    =: \sum_a \tilde{\alpha}_a^i \otimes \tilde{\beta}_a^i \\
    &= \tr\left( \sum_{a'} \tilde{\alpha}_{a'}^i \otimes \tilde{\beta}_{a'}^i\right) \sum_a \frac{\tr\left( \tilde{\alpha}_a^i \otimes \tilde{\beta}_a^i\right)}{\tr\left( \sum_{a'} \tilde{\alpha}_{a'}^i \otimes \tilde{\beta}_{a'}^i\right)} \frac{\tilde{\alpha}_a^i}{\tr\left( \tilde{\alpha}_a^i\right)} \otimes \frac{\tilde{\beta}_a^i}{\tr\left( \tilde{\beta}_a^i\right)} \\
    & =: \tr\left( \sum_{a'} \tilde{\alpha}_{a'}^i \otimes \tilde{\eta}_{a'}^i\right) \sum_a p(a|i) \hat{\tilde{\alpha}}_a^i \otimes \hat{\tilde{\eta}}_a^i \: .
\end{align}
Up to normalisiation, this is a separable state. Therefore, it is in the cone of separable states. 
Hence, the notion presented in Eq.~(\ref{eq:SEP_POVM}) holds if the standard definition is satisfied.

Secondly, assume that the notion from Eq.~(\ref{eq:SEP_POVM}) is satisfied.
For a state $\omega$ to be in the cone of separable states, it means that it has a decomposition of the form $\omega = \tr(\omega) \sum_\lambda p_\lambda \xi_\lambda \otimes \zeta_\lambda$.
As Eq.~(\ref{eq:SEP_POVM}) gives a state in the separable cone for any local operations $\Lambda$ and $\Gamma$, it holds especially for identity channels and we get by choosing invertible states $\alpha$ and $\beta$:
\begin{align}
    M_i
    = \tr(\omega_i) \sum_\lambda p_\lambda^i \alpha^{-\frac{1}{2}} \xi_\lambda^i \alpha^{-\frac{1}{2}} \otimes \beta^{-\frac{1}{2}} \zeta_\lambda^i \beta^{-\frac{1}{2}}
    =: \sum_\lambda (\tr(\omega_i) p_\lambda^i) A_\lambda^i \otimes B_\lambda^i \: .
\end{align}
Therefore, the standard definition is also satisfied.

This means that the two notions defining the separability of measurements are equivalent.

\section{Derivation of the bound on the convex weight for the whole network} \label{app:allPOVM&States}

In the following the bound on the convex weight of the whole network, i.e. all states and measurements, is derived.
Say a line network consists of measurements $M_{b_j}^j$ which according to the definition of the convex weight for measurements can each be decomposed into
\begin{equation}
    M_{b_j}^j = (1 - v_j) P_{b_j}^j + v_j N_{b_j}^j \:\forall b_j,
\end{equation}
where $j$ labels the measurement and $P^j$ is a separable measurement,
and states $\rho_j$ which each have the convex weight decomposition
\begin{equation}
    \rho_j = (1-w_j) \tau_j + w_j \eta_j,
\end{equation}
where $\tau_j$ is a separable state.

First, consider the subnormalised state $\sigma_{b_1}$ prepared by the first measurement $M_{b_1}^1$ and the first two states $\rho_1$ and $\rho_2$
\begin{align}
    \sigma_{b_1} 
    & = (\text{id} \otimes \mathcal{E}_{M^1_{b_1},\rho_2})(\rho_1) \label{eq:sig_b1_POVM}\\
    & = v_1 (\text{id} \otimes \mathcal{E}_{N^1_{b_1},\rho_2})(\rho_1) + (1-v_1) (\text{id} \otimes \mathcal{E}_{P^1_{b_1},\rho_2})(\rho_1) \\
    & = v_1 \left( (w_1 w_2) \eta'_{b_1} + \left(\tr((\text{id} \otimes \mathcal{E}_{N^1_{b_1},\rho_2})(\rho_1))-w_1w_2 \tr(\eta'_{b_1})\right) \hat{\tau}'_{b_1} \right) + (1-v_1) (\text{id} \otimes \mathcal{E}_{P^1_{b_1},\rho_2})(\rho_1) \\
    & =: v_1(w_1 w_2) \eta'_{b_1} + (\tr(\sigma_{b_1})-v_1w_1w_2 \tr(\eta'_{b_1})) \hat{\tilde{\tau}}_{b_1},
\end{align}
where $\sum_{b_1} \tr(\eta'_{b_1}) = 1$, $\hat{\tau}'_{b_1} \in \text{SEP}$, and $\hat{\tilde{\tau}}_{b_1} \in \text{SEP}$.
Note that the third equality holds due to the bound on the convex weight for states in a network, c.f. Appendix~\ref{app:allStates} for a detailed derivation.
Now consider the subnormalised state $\sigma_{b_1,b_2}$ prepared by the first two measurements and the first three states
\begin{align}
    \sigma_{b_1,b_2} 
    = (\text{id} \otimes \mathcal{E}_{M^2_{b_2},\rho_3} \circ \mathcal{E}_{M^1_{b_1},\rho_2}) (\rho_1)
    = (\text{id} \otimes \mathcal{E}_{M^2_{b_2},\rho_3}) (\sigma_{b_1}) \: .
\end{align}
Notice that up to normalisation this is the same form as Eq.~(\ref{eq:sig_b1_POVM}) and hence, 
\begin{align}
    \sigma_{b_1,b_2} 
    =: (v_1v_2)(w_1 w_2 w_3) \tilde{\eta}_{b_1,b_2} + (\tr(\sigma_{b_1,b_2}) -v_1v_2w_1w_2w_3 \tr(\tilde{\eta}_{b_1,b_2})) \hat{\tilde{\tau}}_{b_1,b_2}
\end{align}
where $\sum_{b_1,b_2} \tr(\tilde{\eta}_{b_1,b_2}) = 1$ and $\hat{\tilde{\tau}}_{b_1,b_2} \in \text{SEP}$.
This procedure is applicable to all measurements and states in the network.
Thus, the subnormalised state $\sigma_{\vec{b}}$ prepared by the network takes the form
\begin{align}
    \sigma_{\vec{b}} 
    & = \left(\prod_j v_j\right)\left(\prod_j w_j\right) \tilde{\eta}_{\vec{b}} + \left(\tr(\sigma_{\vec{b}})-\prod_j v_j\prod_j w_j \tr(\tilde{\eta}_{\vec{b}}) \right) \hat{\tilde{\tau}}_{\vec{b}},
\end{align}
where $\sum_{\vec{b}} \tr(\tilde{\eta}_{\vec{b}}) = 1$ and $\hat{\tilde{\tau}}_{\vec{b}} \in \text{SEP}$.
Due to the similarity of this form with the convex weight decomposition, the convex weight of the normalised state $\hat{\sigma}_{\vec{b}}$ prepared by the line network is upper bound by
\begin{equation}
    \text{W}_\mathcal{F}(\hat{\sigma}_{\vec{b}}) \leq \left( \prod_j v_j\right)\left( \prod_j w_j\right) \frac{\tr(\tilde{\eta}_{\vec{b}})}{\tr(\sigma_{\vec{b}})} \: .
\end{equation}
By rearranging and summing over the measurement outcomes $\vec{b}$, the average convex weight bound is obtained.
Note that this bound trivially holds for any subset of measurements and states of the network. A particular instance of this is the bound provided in Observation~\ref{obs:RWstates}.

\bibliographystyle{apsrev4-2}
\bibliography{references}

\end{document}